\documentclass[a4paper,11pt]{article}

\usepackage{amssymb,amsmath,amsfonts,amsthm}
\usepackage{graphicx}
\usepackage{float}
\usepackage{enumerate}

\usepackage{a4wide}

\usepackage{color}
\usepackage[colorlinks=true,hyperindex=true]{hyperref}

\usepackage{authblk}

\newtheorem*{definition}{Definition}
\newtheorem{thm}{Theorem}[section]
\newtheorem{lem}[thm]{Lemma}
\newtheorem{prop}[thm]{Proposition}
\newtheorem{cor}[thm]{Corollary}
\newtheorem*{remark}{Remark}

\title{Exponential Stability of Subspaces\\ for Quantum Stochastic Master Equations}

\author[$\diamondsuit$]{Tristan Benoist\thanks{tristan.benoist@math.univ-toulouse.fr}}
\author[$\clubsuit$]{Cl\'ement Pellegrini\thanks{clement.pellegrini@math.univ-toulouse.fr}}
\author[$\spadesuit$]{Francesco Ticozzi\thanks{ticozzi@dei.unipd.it}}

\affil[$\diamondsuit$]{\small CNRS, Laboratoire de Physique Th\'eorique, IRSAMC}
\affil[$$]{Universit\'e de Toulouse, UPS}
\affil[$$]{F-31062 Toulouse, France}

\affil[$\clubsuit$]{\small Institut de Math\'ematiques de Toulouse}
\affil[$$]{Equipe de Statistique et de Probabilit\'e}
\affil[$$]{Universit\'e Paul Sabatier}
\affil[$$]{31062 Toulouse Cedex 9, France}

\affil[$\spadesuit$]{\small Dept. of Information Engineering}
\affil[$$]{Universit\'a degli Studi di Padova}
\affil[$$]{via gradenigo 6/b, 35131 Padova, Italy}
\affil[$$]{\small and Dept. of Physics and Astronomy}
\affil[$$]{Dartmouth College}
\affil[$$]{6127 Wilder, Hanover, NH (USA)}

\begin{document}
\maketitle
\begin{abstract}
We study the stability of quantum pure states and, more generally, subspaces for stochastic dynamics that describe continuously--monitored systems. We show that the target subspace is almost surely invariant if and only if it is invariant for the average evolution, and that the same equivalence holds for the global asymptotic stability. Moreover, we prove that a strict linear Lyapunov function for the average evolution always exists, and latter can be used to derive sharp bounds on the Lyapunov exponents of the associated semigroup. Nonetheless, we also show that taking into account the measurements can lead to an improved bound on stability rate for the stochastic, non-averaged dynamics. We discuss explicit examples where the almost sure stability rate can be made arbitrary large while the average one stays constant.
\end{abstract}

\section{Introduction}
\textbf{General context:} Pure quantum states play a key role in many aspects of quantum theory, and quantum dynamics in particular: they are associated to eigenstates of Hamiltonians with non-degenerate spectrum, and hence to ground states representing the zero-temperature equilibria for the system; they are the output of measurement processes corresponding to non-degenerate observables; pure states are typically used to represent information in quantum information processing and communication; furthermore, nonclassical correlations in quantum mechanics are best exhibited by maximally entangled states for joint systems, which are pure. This central role motivates a growing interest in characterizing evolutions that converge to specific classes of pure states of interest. 

A similar interest lays on convergence to subspaces of Hilbert space, whether they represent energy eigenspaces, they are associated to certain excitation numbers or symmetric states, or represent the support for a quantum error-correcting code.

In order for a quantum dynamical system to converge to a pure state or a subspace irrespective of the initial state, it needs to include some interaction with its environment, namely it needs to be an {\em open system}.
We shall focus on Markov quantum systems associated to Stochastic Master Equations (SME) and their corresponding semigroups \cite{barchielli2009,alicki-lendi,gardiner-qn}. This class of models emerges naturally in many quantum atomic, optical and nanomechanical systems \cite{qed1,qed2,nanomechanical1}. It is of interest in measurement and decoherence theory \cite{belavkin1985, wiseman-adaptive,adler,pellegrini1,pellegrini2,pellegrini3,bbbqnd,bbbcontinuous}, and it has a central role in quantum filtering and measurement-based feedback control systems \cite{belavkin1992, belavkin-report, wiseman-book,altafini-introduction,amini_rouchon_pellegrini,Rouchon1,Rouchon2,Rouchon3,Rouchon4,pellegrini4}.

In many applications, convergence is not enough: a {\em fast} preparation of the target set needs to be enacted. A fast convergence is also needed to best protect the system from undesirable external perturbations that may be negligible on the state preparation time scale but are relevant on longer time scales, making the preparation robust. This is, for example, crucial towards implementing effective quantum memories \cite{ticozzi-QDS,shabani-init,ticozzi-isometries}. Different ways to characterize the speed of convergence, as well as asymptotic invariant sets, have been developed for Markovian evolutions \cite{baumgartner-2,ticozzi-QDSspeed,cirillo-decompositions}.

In \cite{ticozzi-stochastic}, a general approach to stabilization of {\em diffusive} SME has been proposed, which relies as much as possible on open-loop control and resorts to feedback design only when the open-loop control cannot achieve the desired task. The motivation for this choice is twofold: on the one hand, open-loop control is easier to implement, as it does not require the taxing computational overhead of integrating the SME in real time. On the other hand, simulations showed that the open-loop controlled evolution converged exponentially. This is not completely surprising, as it is in agreement with another result of the paper: convergence {\em in probability} to subspaces for the SME can be proved by checking if the {\em mean} evolution converges to the same subspace. In this paper, we shall make such observation rigorous for a larger, more general class of dynamics and further explore their convergence features, comparing average and almost sure convergence.

\bigskip

\noindent{\textbf{Dynamics of interest:}} We study quantum systems described by a finite dimensional Hilbert space $\cal H$. The possible states of the system are then given by density matrices on $\cal H$. Namely, $\rho$ is a state if and only if it is an element of the set
$$\mathcal S(\mathcal H):=\{\rho\in\mathcal B(\mathcal H)\ |\ \rho\geq 0, \textrm{tr}\rho=1\}$$
where $\mathcal B(\mathcal H)$ is the set of linear operators on $\mathcal H$.

The stochastic dynamics we shall consider are processes $(\rho(t))_{t\in\mathbb R_+}$ of states solving stochastic differential equations of the form:
\begin{eqnarray}\label{eq:def_trajectory0}
   d\rho(t)&=&{\cal L}(\rho(t-))dt + \sum_{i=0}^p {\cal G}_i(\rho(t-)) dW_i(t)\nonumber\\&& + \sum_{i=p+1}^n \left(\frac{{\cal J}_i(\rho({t-}))}{\textrm{tr}[\mathcal J_i(\rho({t-}))]}-\rho({t-})\right)(d\hat{N}_i(t)-\textrm{tr}[\mathcal J_i(\rho({t-}))]dt),
\end{eqnarray}
including both diffusive ($dW_t$) Wiener processes and Poisson, or ``jump'', processes ($d\hat N_t$) with intensity $\textrm{tr}[\mathcal J_i(\rho({t-}))]dt$.
The operators appearing in \eqref{eq:def_trajectory0} are defined by
 \begin{equation}\label{eq:def_functions}
   \begin{split}
    {\cal L}(\rho)&=-i[H,\rho] + \sum_{i=0}^n\left( C_i\rho C_i^* - \frac12 \big(C_i^*C_i \rho + \rho C_i^*C_i\big)\right),\\
   {\cal J}_i(\rho)&=C_i\rho C_i^*,\,\, i=0,\ldots,n,\\
  {\cal G}_i(\rho)&=C_i\rho + \rho C_i^* - {\rm tr}[(C_i + C_i^*)\rho]\rho,\,\, i=0,\ldots,n,
   \end{split}
  \end{equation}
  where $C_i,i=0,\ldots,n$ are elements of $\mathcal B(\mathcal H)$, and where $H$ is a self-adjoint element of $\mathcal B(\mathcal H)$. 


%
An equation of the form \eqref{eq:def_trajectory0} defines a generic evolution of quantum system undergoing continuous indirect measurements. It is called SME \cite{barchielli2009}, or filtering equation, in the control-oriented community \cite{belavkin1992,altafini-introduction,amini_rouchon_pellegrini, Rouchon1,Rouchon2}. Its solution $(\rho(t))_t$ is called a \emph{quantum trajectory}.

The class of evolutions captured by \eqref{eq:def_trajectory0} comprises all evolution of a system (an atom or a spin) interacting with a electromagnetic field which is monitored \cite{barchielli2009,qed1} as well as nano-mechanical devices \cite{nanomechanical1}. Hence, these include the typical models used for (measurement-based) feedback stabilization \cite{altafini-introduction}. Similar models can also be derived for discrete-time evolutions, and have received particular attention given their applicability to current experimental setups \cite{vanhandel-invitation,harochebook}. In the continuous time limit these discrete models converge weakly to solutions of SME \cite{pellegrini1,pellegrini2,pellegrini3,bbbqnd,bbbcontinuous}. 
Physically, $H$ corresponds to the effective Hamiltonian for the system which includes its internal Hamiltonian and a perturbation (Lamb shift) induced by the interaction with its environment. The environment is typically associated to a number of quantum fields, and the interaction of the system with the latter is described by the operators $C_i$. 

With respect to \cite{ticozzi-stochastic} or most control-oriented work, in \eqref{eq:def_trajectory0} we consider both processes $(W_i(t))$ corresponding to diffusive evolutions and $(\hat{N}_i(t))$ corresponding to Poisson processes (jump processes) with stochastic intensity. These canonical stochastic processes represent the fluctuations of the outcome of continuous measurements performed on the fields, after their interaction with the system.  Poisson processes are associated to particle  counting measurements (typically photons), whereas Wiener processes are associated to particle currents or field quadrature measurements \cite{barchielli2009,belavkin1992,pellegrini1,pellegrini3,bbbqnd,bbbcontinuous}.

The operator $\mathcal L$ in \eqref{eq:def_functions} has the form of so-called Lindblad operators \cite{lindblad1976,gorini1978}, namely the generator of a semi group of completely positive, trace preserving (CPTP) maps on the set of states $\mathcal S(\mathcal H)$. These generators correspond to master equations for open quantum systems \cite{alicki-lendi,petruccione}, and have been extensively studied. Being linear deterministic systems on a convex, positive set the study of their properties is generally simpler than studying directly the stochastic evolution. Their stability and controllability properties are discussed for example in \cite{ticozzi-QDS, ticozzi-markovian, ticozzi-QDSspeed, thomas-liesemi group, altafini-markov}. 

In our case, ${\cal L}$ is also the generator of the Markov semigroup associated to the stochastic system. These represent the best description of the state evolution when the measurement record is not accessible, and can thus be obtained as the expectation of \eqref{eq:def_trajectory0} over the outcomes of the  measurement processes. Namely, if $$\hat\rho(t)=\mathbb E[\rho(t)],$$
 it follows from \eqref{eq:def_trajectory0} that,
 $$\frac{d}{dt}\hat\rho(t)=\mathcal L(\hat\rho(t)).$$ 
In this work, we shall exploit known and new properties of the semigroup evolution to obtain new results regarding the stochastic ones.
 
 \noindent{\textbf{Main results:}} The principal aim of this paper is to study asymptotic stability of $(\rho(t))$ towards attracting subspaces, as well as to provide sharp bounds on its rate of convergence. Let $\mathcal H_S$ be the target subspace of $\mathcal H$. The whole Hilbert space can be decomposed in the direct sum $\mathcal H=\mathcal H_S\oplus\mathcal H_R,$ where $\mathcal H_R$ corresponds to the orthogonal complement of $\mathcal H_S$. Denoting $P_S$ the orthogonal projector on $\mathcal H_S$ the following set \[\mathcal I_S(\mathcal H)=\{\rho\in \mathcal S(\mathcal H)\mbox{ s.t. }{\rm tr}(P_S\rho)=1\}.\] represents the set of states whose support is $\mathcal H_S$ or a subspace of $\mathcal H_S$. When we are concerned with pure state preparation, we have $\mathcal H_S=\mathbb C |\phi\rangle,$ with $|\phi\rangle$ the pure state to be prepared. The following definition addresses the invariance and asymptotic attractivity properties of interest.
 
\begin{definition}
The subspace $\mathcal H_S$ is said invariant 
\begin{itemize}
\item in mean if
\[\rho_0\in\mathcal I_S(\mathcal H)\Rightarrow \; \hat\rho(t)\in \mathcal I_S(\mathcal H),\quad \forall t>0.\]
\item almost surely if
\[\rho_0\in\mathcal I_S(\mathcal H)\Rightarrow  \rho(t)\in \mathcal I_S(\mathcal H),\;\forall t>0 \quad a.s  .\]
\end{itemize}

The subspace $\mathcal H_S$ is said globally asymptotic (GAS) 
\begin{itemize}
\item in mean if $\forall \rho_0\in\mathcal S(\mathcal H)$,
\[\lim_{t\to\infty}\left\|\hat\rho(t) - P_S\hat\rho(t)P_S\right\|=0.\]
\item almost surely if $\forall \rho_0\in\mathcal S(\mathcal H)$,
\[\lim_{t\to\infty}\left\|\rho(t) - P_S\rho(t)P_S\right\|=0\quad a.s .\]
\end{itemize}
\end{definition}
Stability of pure states and subspaces for CPTP map semigroups has been discussed in \cite{ticozzi-markovian,ticozzi-QDS,ticozzi-QDSspeed,baumgartner-2}. In particular, it is proven that the Lindblad operators must exhibit a particular structure in order to ensure mean invariance and mean GAS.  Building on this framework, we obtain the following theorem:
\begin{thm}[Invariance and stability in mean {\em iff} almost sure]\label{thm-mean-as}
The subspace $\mathcal H_S$ is invariant in mean if and only if it is invariant almost surely.
The space $\mathcal H_S$ is GAS in mean if and only if it is GAS almost surely.
\end{thm}
A similar result for SME including only diffusive terms has been estabilished in  Theorem 3.1 and Proposition 4.1 of \cite{ticozzi-stochastic}, focusing on the relation between invariance and convergence in mean with  invariance and convergence \textit{in probability}. Here we obtain results for general SME of type \eqref{eq:def_trajectory0} in an {\em almost sure} sense, which is stronger, and provide a more direct proof. 

We next establish that GAS subspace for the dynamical evolutions of interest are in fact also exponentially stable by using Lyapunov function techniques. The second main result proves that asymptotic stability of a subspace is equivalent to, and not just implied by, the existence of a {\em linear} Lyapunov function for the semigroups of interest.
\begin{thm}[Linear Lyapunov Function for GAS Subspaces]\label{thm:convlyap}
A subspace ${\cal H}_S$ is GAS if and only if there exists  a linear function $V_K:\mathcal S(\mathcal H)\to\mathbb R_+$ such that:
\begin{align}
&V_K(\rho)\geq 0,\quad \textrm{with }\; V_K(\rho)= 0\;\textrm{if and only if }\; \rho\in{\cal I}_S({\cal H});\\
&\label{eq:pospos}V_K({\cal L}(\rho))<0\quad \textrm{for all}\; \rho\notin{\cal I}_S({\cal H}).
\end{align}
\end{thm}
The result can thus be seen as a {\em converse} Lyapunov Theorem, which is of practical interest in many situations in which one would like to prove that a given controlled dynamics converges to a target pure state, as well as to develop insights in design methods for dissipative quantum control \cite{poyatos,ticozzi-QDSspeed,ticozzi-QL1,scaramuzza-switching}.
Beside its own relevance, the above result is going to be instrumental in deriving our bounds on the convergence speed, which are summarized in our third main result.
\begin{thm}[Lyapunov exponents for GAS subspaces] \label{thm:expo_conv} Assume $\mathcal H_S$ is GAS and denote $V(\rho)=\mathrm{tr}(P_R\rho)$, where $P_R$ is the orthogonal projector on $\mathcal H_R$. Then there exists $\alpha_0>0$ such that
\begin{align}
\limsup_{t\rightarrow\infty}\frac{1}{t}\ln(V(\hat\rho(t)))\leq-\alpha_0\label{eq:mean_mean_rate_conv},
\end{align}
\begin{align}
\limsup_{t\rightarrow+\infty} \frac{1}{t}\ln(V(\rho(t))\leq -\alpha_0\quad \mbox{a.s.}.\label{eq:as_mean_rate_conv}
\end{align}
Moreover, if the condition $\mathbf{(SP)}$ $P_RC_{j}^*P_RC_{j}P_R>0$ for all $j=p+1,\ldots,n$ holds, then there exists $\beta_0\geq\alpha_0$ such that
\begin{align}
\limsup_{t\rightarrow+\infty} \frac{1}{t}\ln(V(\rho(t))\leq -\beta_0\quad \mbox{a.s.}\label{eq:as_as_rate_conv}
\end{align}
\end{thm}
This result shows that the exponential stability in mean and the almost sure one are comparable if $\mathcal H_S$ is GAS. In particular, the same bound for the Lyapunov exponent holds. Moreover, under assumption {\bf SP}, explicitly considering the measurements can lead to an improved exponential stability. Actually, in Section \ref{sec:simulations} we show that adding a specific extra indirect measurement, one can arbitrary increase the stability rate for the stochastic dynamics while keeping the same stability for the average dynamics. That is, one can taylor an experiment with $\alpha_0$ fixed and $\beta_0$ arbitrary large. 

Next corollary elucidates further the convergence towards $\mathcal H_S$.
\begin{cor}\label{thm:littleo}
Assume $\mathcal H_S$ is GAS and let $\beta_0\geq\alpha_0>0$ be the same as in Theorem \ref{thm:expo_conv}. Then for any $\epsilon>0$,
\begin{align}\label{eq:littleo_mean}
\left\|\rho(t)-P_S\rho(t)P_S\right\|=&o(e^{-\frac12(\alpha_0-\epsilon)t})\quad \mbox{a.s. and in $L^1$--norm}
\end{align}
and if moreover assumption {\bf SP} is fulfilled,
\begin{align}\label{eq:littleo_as}
\left\|\rho(t)-P_S\rho(t)P_S\right\|=&o(e^{-\frac12(\beta_0-\epsilon)t})\quad a.s.
\end{align}
\end{cor}
\bigskip

The remainder of paper consist essentially of the proofs of these results. It is structured as follows. Section 2 is devoted to recalling and deriving new properties of the deterministic semigroup dynamics. We first recall the result of \cite{ticozzi-markovian,ticozzi-QDS}. In particular, these give the explicit structure of the Lindblad operator that guarantees that a subspace is invariant and GAS. Next we prove Theorem \ref{thm:convlyap}, namely that the Perron--Frobenius Theorem for completely positive evolutions \cite{Evans} can be used to systematically derive a {\em linear} Lyapunov function that shows that a subspace is GAS.  As a corollary, we can directly prove the first bound in Theorem \ref{thm:expo_conv}. In Section 3, we present the probabilistic setting needed to formally introduce SME and quantum trajectories. Section 4 is mainly dedicated to the improved bound for stochastic exponential stability, which is the content of the second part of Theorem \ref{thm:expo_conv}. Finally in Section 5, we provide an example in which we can increase arbitrarily the stability rate for the almost sure convergence, while leaving the one for the average one invariant.

\section{Deterministic Result--A Converse Lyapunov Theorem}

This section concerns deterministic results. That is results for the mean of $(\rho(t))$. Let us first recall the result of the general structure of the Lindblad operator $\mathcal L$ leading to invariance and GAS \cite{ticozzi-markovian,ticozzi-QDS}. The definition of $\mathcal H_S$ and $\mathcal H_R$ allows for a convenient decomposition of all the matrices.  Let $X\in \mathcal B(\mathcal H)$, then its matrix representation in an appropriately chosen basis can be written as
\[X=\left(\begin{array}{cc}X_S&X_P\\X_Q&X_R\end{array}\right),\]
where $X_S,X_R,X_P,X_Q$ are operators from $\mathcal H_S$ to $\mathcal H_S$, from $\mathcal H_R$ to $\mathcal H_R$, from $\mathcal H_R$ to $\mathcal H_S$ and from $\mathcal H_S$ to $\mathcal H_R$, respectively.
In the rest of the paper, the indexes $S,R,P,Q$ will refer to the same blocks as above.

The invariance and GAS properties in mean are directly related to the Jordan structure and/or irreducibility of the completely positive map semi group $e^{t\mathcal L}$. We here recall the relevant results without proof and refer the interested reader to the original articles .
\begin{thm}[\cite{ticozzi-QDS}]\label{thm:mean inv condition}
The subspace $\mathcal H_S$ is invariant in mean if and only if 
\[\forall j, C_{j,Q}=0\mbox{ and } iH_P-\frac12\sum_j C_{j,S}^*C_{j,P}=0.\]
\end{thm}
\begin{thm}[\cite{ticozzi-markovian}]
The subspace $\mathcal H_S$ is GAS in mean if and only if no invariant subspaces are included in $\bigcap_{j}\ker(C_{j,P})$.
\end{thm}

With these results in mind we turn to the proofs of Theorem \ref{thm:convlyap} and \eqref{eq:mean_mean_rate_conv} in Theorem \ref{thm:expo_conv}. A key tool in deriving stability rates of the GAS subspace $\mathcal H_S$ is the construction of a suitable {\em linear Lyapunov function} for the corresponding semi group evolution. While typically this is not possible for linear systems, where the natural Lyapunov functions are quadratic, in this case we can exploit: (i) the positivity of the evolution, so that a Perron--Frobenius type result holds; and (ii) the fact that the stable set has support on a subspace of ${\cal H}$.

Let us recall some well known facts on semi groups of completely positive maps. 
A continuous semi group on $\mathcal{B}({\mathcal H })$ is completely positive if and only if its generator $\cal K$ has the form \cite{lindblad1976,GKS}:
\begin{equation}\label{CPform}
\mathcal K(X)=G^*X+XG+\Psi(X),
\end{equation}
where $\Psi(X)$ is a completely positive map from $\mathcal B(\mathcal H)$ to itself and $G$ is an element of $\mathcal B(\mathcal H)$. A positive linear map from $\mathcal B(\mathcal H)$ to itself is called irreducible if it does not admit nontrivial invariant subspaces or, equivalently, invariant operators are full rank. A generator is said to be irreducible if the semi group it generates is of irreducible maps.

The following Lemma gives a sufficient condition on $\Psi$ such that $\mathcal K$ generates a semi group of irreducible completely positive maps. It is a weaker version of \cite[Theorem 2.3]{jaksic}. We reproduce the proof from \cite{jaksic} for the reader convenience.
\begin{lem}\label{lem:irr from Psi}
Let $\mathcal K X=G^*X+XG+\Psi(X)$ be the generator of a semi group of completely positive maps on $\mathcal B(\mathbb C^d)$, $d\in\mathbb N$. If $\Psi$ is irreducible, then  $e^{t\mathcal K}$ is irreducible $\forall t\in\mathbb R_+.$
\end{lem}
\begin{proof}
The proof provides actually a stronger result. Namely, it shows that for any nonzero $|\phi\rangle,|\psi\rangle\in\mathbb C^d$, for any $t>0$, $\langle \psi|e^{t\mathcal K}(|\phi\rangle\langle\phi|)\psi\rangle>0$. This property is called positivity improving in \cite{jaksic}.
First, from \cite[Lemma 2.1]{Evans}, $\Psi$ irreducible implies that
\[\langle \psi|({\rm Id}+\Psi)^{d-1}(|\phi\rangle\langle\phi|)\psi\rangle>0\]
for any nonzero $\phi,\psi\in \mathbb C^d$.
Making an expansion of both $e^{t\Psi}$ and $({\rm Id}+\Psi)^{d-1}$ one see that all the terms are positive, and all the terms in the second expansion also appear in the first one. Hence, for any $t>0$, there exists $c>0$ such that $e^{t\Psi}\geq c ({\rm Id}+\Psi)^{d-1}$. Therefore, $e^{t\Psi}$ is positivity improving.

Now notice that that $e^{t\mathcal K_0}:X\mapsto e^{tG^*}Xe^{tG}$ is a semi group of completely positive maps. We define the family of completely positive maps: 
\[\Gamma_t(X)=e^{-t\mathcal K_0}e^{t\mathcal K}X.\]
Since for any $t>0$ and $|\phi\rangle\in\mathbb C^d$, $|\phi\rangle\neq0$, $e^{tG}|\phi\rangle\neq0$, it remains to show that for any $t>0$, $|\phi\rangle,|\psi\rangle\in\mathbb C^d$, $|\phi\rangle\neq0,|\psi\rangle\neq0$, $\langle \psi|\Gamma_t(|\phi\rangle\langle\phi|)\psi\rangle>0$.

Suppose $\langle \psi|\Gamma_{t_0}(|\phi\rangle\langle\phi|)\psi\rangle=0$ for a fixed $t_0$. The Dyson expansion of $\Gamma_{t_0}$ is
\[\Gamma_{t_0}={\rm Id} +\sum_n \int_{0<s_1<\ldots<s_n<t_0} \Psi_{s_1}\circ\cdots\circ\Psi_{s_n} ds_1\ldots ds_n\]
where $s\rightarrow\Psi_s=e^{-s\mathcal K_0}\circ \Psi \circ e^{s\mathcal K_0}$ is a family of continuous completely positive maps. It follows
\begin{align*}
\langle\psi|\Gamma_{t_0}(|\phi\rangle&\langle\phi|)\psi\rangle=\\
&|\langle\psi|\phi\rangle|^2 +\sum_n \int_{0<s_1<\ldots<s_n<t_0} \langle \psi|\Psi_{s_1}\circ\cdots\circ\Psi_{s_n}(|\phi\rangle\langle\phi|)\psi\rangle ds_1\ldots ds_n.
\end{align*}
All the integrands are positive and continuous in $s_1,\ldots,s_n$. Hence the assumption $\langle\psi|\Gamma_{t_0}(|\phi\rangle\langle\phi|)\psi\rangle=0$ implies $\langle \psi|\Psi_{s_1}\circ\cdots\circ\Psi_{s_n}(|\phi\rangle\langle\phi|)\psi\rangle=0$ for all $(s_1,\ldots,s_n)\in[0,t_0]^n$. Especially $\langle \psi|\Psi^{\circ n}(|\phi\rangle\langle\phi|)\psi\rangle=0$ for all $n\in\mathbb N\setminus\{0\}$. It follows that for all $t>0$, $\langle \psi|e^{t\Psi}(|\phi\rangle\langle\phi|)\psi\rangle=0$. This implies that either $\phi$ or $\psi$ must be $0$. Hence for all non zero $\phi$ and $\psi$, and for all times $t$, $\langle\psi|\Gamma_{t}(|\phi\rangle\langle\phi|)\psi\rangle>0$ thus, setting $\psi=e^{tG}\psi'$ one obtain that for any $t$,
\[\langle \psi'|e^{t\mathcal K}(|\phi\rangle\langle\phi|)\psi'\rangle>0.\]
The map $e^{t\mathcal K}$ is positivity improving and therefore irreducible. 
\end{proof}

We turn to the finer structure of the semi group $t\mapsto e^{t\mathcal L}$ implied by the stability of $\mathcal H_S$. Let $\rho_S\in \mathcal B(\mathcal H_S)$ and $\rho_R\in\mathcal B(\mathcal H_R)$ such that $\rho_S\geq 0$ and $\rho_R\geq 0,$ but not necessarily of trace one. Define, using the block--decomposition with respect to the orthogonal direct sum decomposition $\mathcal H=\mathcal H_S\oplus \mathcal H_R$ introduced before, the maps:
\begin{align*}
{\cal L}_S(\rho_S)=&-i[H_S,\rho_S] + \sum_j C_{j,S}\rho_SC_{j,S}^* -\frac12\{C_{j,S}^*C_{j,S},\rho_S\},\\
{\cal L}_R(\rho_R)=&-i[H_R,\rho_R] + \sum_j C_{j,R}\rho_RC_{j,R}^* -\frac12\{C_{j,P}^*C_{j,P}+C_{j,R}^*C_{j,R},\rho_R\}.
\end{align*}
Then, the following Proposition holds.
\begin{prop}\label{prop:LR}
The family $\{e^{t\mathcal L_S}\}_{t\geq 0}$ is a semi group of trace preserving completely positive maps, and $\{e^{t\mathcal L_R}\}_{t\geq 0}$ is a semi group of trace non increasing completely positive maps.
\end{prop}
\begin{proof}
Both generators have the form \eqref{CPform} and thus generate semigroups of completely positive maps.

The operator $\mathcal L_S$ have the form of a Lindblad operator, thus $\textrm{tr}[\mathcal L_S(\rho_S)]=0$ for any $\rho_S\in\mathcal B(\mathcal H_S)$ and $\{e^{t\mathcal L_S}\}_{t\geq 0}$ is trace preserving.

Since for any $t\geq 0$, $e^{t\mathcal L_R}$ is a positive map, we have $\rho_R\geq 0\Rightarrow e^{t\mathcal L_R}\rho_R\geq 0$. Moreover,
$$e^{t\mathcal L_R}\rho_R=\rho_R+\int_0^t \mathcal L_Re^{s\mathcal L_R}\rho_R ds$$
and ${\rm tr}(\mathcal L_R\rho_R)=-\sum_j {\rm tr}(C_{j,P}^*C_{j,P}\rho_R)\leq0$ for any positive semi definite $\rho_R\in \mathcal B(\mathcal H_R)$.
Thus $e^{t\mathcal L_R}$ is trace non-increasing for all $t\in \mathbb R_+$.
\end{proof}
\noindent The following proposition clarifies the signification of the semigroups we just defined. Recall the averaged evolution is given by
\[\hat \rho(t)=e^{t\mathcal L}\rho(0),\]
where ${\cal L}$ has the form given in \eqref{eq:def_functions}.
\begin{prop}
Assume $\mathcal H_S$ is invariant.
If $\rho\in\mathcal I_S(\mathcal H)$, then
\begin{equation}e^{t\mathcal L}\rho=\left(\begin{array}{cc}e^{t\mathcal L_S} \rho_S&0\\0&0\end{array}\right)\label{eq:L_S semi group}\end{equation}
and for any $\rho(0)\in\mathcal S(\mathcal H)$, the $R$-block of $\hat\rho(t)$ is
\begin{equation}\hat\rho_R(t)=e^{t\mathcal L_R}\rho_R(0).\label{eq:L_R semi group}\end{equation}
\end{prop}
\begin{proof}
Assuming $\rho\in\mathcal I_S(\mathcal H)$, the invariance of $\mathcal H_S$ implies $e^{t\mathcal L}\rho\in\mathcal I_S(\mathcal H)$. From the invariance condition of Theorem \ref{thm:mean inv condition}, it follows that for any $\rho\in\mathcal I_S(\mathcal H)$,
\[\mathcal L\rho=\left(\begin{array}{cc}\mathcal L_S\rho_S&0\\0&0\end{array}\right).\]
Thus $\hat\rho_S(t)$ is the unique solution of $\frac{d\hat\rho_S(t)}{dt}=\mathcal L_S\hat\rho_S(t)$ which is $e^{t\mathcal L_S}\rho_S$.
The invariance condition of Theorem \ref{thm:mean inv condition} gives immediately
\[\frac{d\hat\rho_R(t)}{dt}=\mathcal L_R\hat\rho_R(t)\]
and the result follows from the uniqueness of the solution.
\end{proof}

The next Lemma is the key one. Let us denote the spectral abscissa of ${\cal L}_R$ as:
\begin{equation} \alpha_0:=\min\{-{\rm Re}(\lambda)\,|\,\lambda\in{\rm sp}(\mathcal L_R)\}.\end{equation}
Building on the Perron--Frobenius Theorem for completely positive maps \cite{Evans}, the fact that ${\cal L}_R$ is a trace-nonicreasing CP generator (see Proposition \ref{prop:LR}) implies that there exists a corresponding positive semi-definite eigenoperator. We here show that there exists an arbitrary small perturbation of the generator for which such operator $K_R$ is actually {\em positive}.

\begin{lem}\label{lem:spectral_radius_majoration}
For any $\epsilon>0$ there exists $K_R>0$ such that
\[\mathcal L_R^*(K_R)\leq -(\alpha_0 -\epsilon)K_R\]
where $\mathcal L^*_R$ is the adjoint of $\mathcal L_R$ with respect to the Hilbert--Schmidt inner product on $\mathcal B(\mathcal H_R)$.
\end{lem}
\begin{proof}
By definition, for any $t\in \mathbb R_+$, the spectral radius of $e^{t\mathcal L_R}$ is $e^{-\alpha_0 t}$. If the completely positive maps of the semi group $e^{t\mathcal L_R}$ are irreducible the existence of $K_R>0$ follows directly from Perron--Frobenius Theorem \cite{Evans}. Indeed it implies there exists $ K_R>0$ such that $\mathcal L_R^*K_R=-\alpha_0 K_R$.

If $\mathcal L_R$ generate a semi group of reducible completely positive maps such positive definite $K_R$ may not exist but a perturbed generator will generate irreducible maps. Let $\Psi:\mathcal B(\mathcal H_R)\to\mathcal B(\mathcal H_R)$ be an irreducible completely positive map. From Lemma \ref{lem:irr from Psi}, it follows that for all $\eta>0$, $\mathcal L_\eta^*=\mathcal L_R^*+\eta \Psi$ is the generator of a semi group of irreducible completely positive maps. Let $\alpha_\eta=\min\{-{\rm Re}(\lambda)\,|\,\lambda\in{\rm sp}(\mathcal L_\eta^*)\}$. From Perron--Frobenius Theorem \cite{Evans}, there exists $K_\eta>0$ such that $\mathcal L_\eta^* K_\eta=-\alpha_\eta K_\eta$. 

Since $\lim_{\eta\to 0}\mathcal L_\eta^*=\mathcal L_R^*$, we have $\lim_{\eta\to 0}\alpha_\eta=\alpha_0$. Hence for any $\epsilon>0$ there exists a $\eta$ small enough such that $\alpha_\eta\geq\alpha_0-\epsilon$. Thus $\mathcal L_R^*K_\eta=-\alpha_\eta K_\eta - \eta \Psi(K_\eta)\leq -(\alpha_0-\epsilon)K_\eta -\eta \Psi(K_\eta)$. Since $\Psi$ is positive and $K_\eta>0$, $\mathcal L_R^* K_\eta\leq -(\alpha_0-\epsilon)K_\eta$ and the result follows setting $K_R=K_\eta$.
\end{proof}

In the construction of our Lyapunov function we shall need the following notation. We extend any linear operator $K_R$ on $\mathcal H_R$ to a linear operator on $\mathcal H$ by putting \[K=\left(\begin{array}{cc}0&0\\0&K_R\end{array}\right).\]
For any such extension $K$ of an operator $K_R>0$, let
\begin{align}
V_K:\ &\mathcal S(\mathcal H)\to[0,1]\label{eq:VK}\\
	&\rho\mapsto{\rm tr}(K\rho)={\rm tr}(K_R\rho_R).\nonumber
\end{align}.

\noindent Being ${\cal L}_R$ trace-nonincreasing implies that its spectrum has a negative-semidefinite real part. We next show, using again a Perron-Frobenius approach, that strict negativity, or equivalently positivity of $\alpha_0,$ is equivalent to GAS.
\begin{lem}\label{lem:alpha0pos}
 The subspace $\mathcal H_S$ is GAS, if and only if
 $$\alpha_0>0.$$
\end{lem}

\begin{proof}
Let us first prove that if $\mathcal H_S$ is GAS, we have $\alpha_0>0$. Assume $\alpha_0\leq0$. From Perron--Frobenius Theorem \cite{Evans} there exists $\mu\in\mathcal S(\mathcal H)$ such that $\mu_R\neq 0$ and $e^{t\mathcal L_R}\mu_R=e^{-t\alpha_0}\mu_R$. It follows $V(\hat\mu(t))=e^{-\alpha_0 t}V(\mu)\geq V(\mu)$ for all $t\in\mathbb R_+$. That contradicts the GAS assumption $\lim_{t\to \infty} V(\hat \mu(t))=0$. Hence $\alpha_0>0$.

Concerning the other implication, assume $\alpha_0>0$ and fix $\epsilon$ such that $\alpha_0>\epsilon>0$. Then, by Lemma 3, there exist $K_R\geq I_{\mathcal H_R}$ such that $\mathcal L^*_RK_R\leq -(\alpha_0-\epsilon) K_R$. Using Gronwall's inequality, we get $V_K(\hat\rho(t))\leq e^{-(\alpha_0-\epsilon)t}V_K(\rho_0)$. Since $K_R\geq I_{\mathcal H_R}$, $V(\hat\rho(t))\leq e^{-(\alpha_0-\epsilon)t}V_K(\rho_0)$. It follows that $\mathcal H_S$ is GAS.
\end{proof}

We are now ready to prove Theorem \ref{thm:convlyap} and Theorem \ref{thm:expo_conv} Equation \eqref{eq:mean_mean_rate_conv}.
\begin{proof}[Proof of Theorem \ref{thm:convlyap}]
The ``if'' implication is a direct application of Krasovskii-LaSalle invariance principle. Let us focus on the converse implication. From Lemma \ref{lem:alpha0pos} we can choose a strictly positive operator $K_R$ on $\mathcal H_R$ fulfilling Lemma \ref{lem:spectral_radius_majoration} with $\epsilon=\alpha_0/2$. 
We then clearly have $V_K(\rho)\geq 0,$ and equal to zero if and only if $\rho\in{\cal I}_S({\cal H})$ by construction.
If we compute $V_K({\cal L}(\rho)),$ with $\rho\notin {\cal I}_S({\cal H}),$ we get:
\begin{align*}
V_K({\cal L}(\rho))&={\rm tr}(K{\cal L}(\rho))={\rm tr}(K_R{\cal L}_R(\rho_R))\\
&={\rm tr}(\mathcal L_R^*(K_R)\rho_R)\\
&\leq -\alpha_0/2\;{\rm tr}(K_R\rho_R)\\
&<0.\end{align*}
\end{proof}

\begin{proof}[Proof of Theorem \ref{thm:expo_conv} Equation \eqref{eq:mean_mean_rate_conv}]
From Lemma \ref{lem:spectral_radius_majoration}, there exists $K_R\geq I_{\mathcal H_R}$ such that $V(\hat\rho(t))\leq e^{-(\alpha_0-\epsilon)t}V_K(\rho(0))$. This way for all $\epsilon>0$ and for all $t>0$
$$\frac{1}{t}\ln(V(\hat\rho(t)))\leq -(\alpha_0-\epsilon)+\frac{\ln(V_K(\rho(0)))}{t}$$
Taking first the limsup and then $\epsilon$ goes to zero yields Equation \eqref{eq:mean_mean_rate_conv}.
\end{proof}

\section{Invariant and Stable Subspaces for Quantum Stochastic Master Equations}

\subsection{Probability spaces and stochastic processes}
This section is devoted to the formal introduction of the stochastic models of interest. Consider a filtered probability space $(\Omega,\mathcal F,(\mathcal F_t),\mathbb P)$ satisfying the usual conditions \cite{protter}. Let $(W_j(t)),j=0,\ldots,p$ be standard independent Wiener processes and let $(N_j(dx, dt)),j=p+1,\ldots,n$ be independent adapted Poisson point processes of intensity $dxdt$; the $N_j$'s are independent of the Wiener processes. We assume that $(\mathcal F_t)$ is the natural filtration of the processes $W,N$ and that $\displaystyle \mathcal F_\infty=\bigvee_{t>0}\mathcal F_t=\mathcal F$.

On $(\Omega,\mathcal F,(\mathcal F_t),\mathbb P), $we consider the following stochastic differential equation.
 
 \begin{equation}\label{eq:def_trajectory}
    \begin{split}
    \rho(t)=\rho_0 &+ \int_0^t {\cal L}(\rho(s-))ds \\
		   &+ \sum_{i=0}^p\int_0^t {\cal G}_i(\rho(s-)) dW_i(s) \\
		   &+ \sum_{i=p+1}^n \int_0^t \int_{\mathbb R} \left(\frac{{\cal J}_i(\rho({s-}))}{\textrm{tr}[\mathcal J_i(\rho({s-}))]}-\rho({s-})\right)\mathbf{1}_{0<x<\textrm{tr}[\mathcal J_i(\rho({s-}))]} [N_i(dx,ds)-dxds].
    \end{split}
   \end{equation}
General results of existence and uniqueness of the solution of \eqref{eq:def_trajectory} can be found in \cite{pellegrini1,pellegrini2,pellegrini3,barchielli2009,barchielliholevo}. From Eq. \eqref{eq:def_trajectory}, we introduce the measurement record for counting processes: 
$$\hat{N}_i(t)=\int_0^t \int_{\mathbb R}\mathbf{1}_{0<x<\textrm{tr}[\mathcal J_i(\rho({s-}))]} N_i(dx,ds), i=p+1,\ldots,n.$$
These processes are counting processes with stochastic intensity given by
 $$\int_0^t \textrm{tr}[\mathcal J_i(\rho({s-}))]ds,i=p+1,\ldots,n.$$
 In particular, for any $i\in\{p+1,\ldots, n\}$, the process $(\hat{N}_i(t)-\int_0^t \textrm{tr}[\mathcal J_i(\rho({s-}))]ds)$ is a $(\mathcal{F}_t)$ martingale under the probability $\mathbb P$.

Using the definition of $\hat{N}_i(t),$ we recover Equation \eqref{eq:def_trajectory0} from the Introduction:  
\begin{eqnarray}\label{eq:def_trajectory2}
   d\rho(t)&=&{\cal L}(\rho(t-))dt + \sum_{i=0}^p {\cal G}_i(\rho(t-)) dW_i(t)\nonumber\\&& + \sum_{i=p+1}^n \left(\frac{{\cal J}_i(\rho({t-}))}{\textrm{tr}[\mathcal J_i(\rho({t-}))]}-\rho({t-})\right)(d\hat{N}_i(t)-\textrm{tr}[\mathcal J_i(\rho({t-}))]dt),
\end{eqnarray}

\subsection{Invariant and Stable Subspaces - Proof of Theorem \ref{thm-mean-as}}
The key object in the proof of Theorem \ref{thm-mean-as} is the Lyapunov function:
\begin{align*}
V:&\mathcal S(\mathcal H)\to [0,1]\\
  &\rho\mapsto {\rm tr}(P_R \rho).
\end{align*}
Its relationship with the definition of $\mathcal I_S(\mathcal H)$ is given by the following Lemma (see also e.g. \cite{ticozzi-stochastic}).
\begin{lem}\label{lem:V=0}
\[V(\rho)=0\Leftrightarrow \rho\in\mathcal I_S(\mathcal H),\]
and the process $(V(\rho(t))$ is a positive super martingale.
\end{lem}
\begin{proof}
The equivalence, as well as the positivity of $V$ are immediate consequences definition of $V$ and the fact that $\rho\geq 0$ for any $\rho\in\mathcal S(\mathcal H)$.

For the super martingale property, using the expression \eqref{eq:def_trajectory2} we get for all $t\geq s\geq 0$
\begin{eqnarray*}
\mathbb E(V(\rho(t))\vert\mathcal F_s)=V(\rho(s))+\int_s^t\mathbb E(V(\mathcal L(\rho(u))|\mathcal F_s)du.
\end{eqnarray*}
Now explicit computations give
$$V(\mathcal L(\rho))={\rm tr}[P_R \mathcal L\rho]=-\sum_j {\rm tr}[{C_{j,P}}^*C_{j,P}\,\,\rho_R]\leq 0,$$
for all $\rho\in\mathcal S(\mathcal H)$.
This way, for all $t\geq s\geq0$
\begin{eqnarray*}
\mathbb E(V(\rho(t))\vert\mathcal F_s)\leq V(\rho(s)),
\end{eqnarray*}
which corresponds to the super martingale property.
\end{proof}

We are now ready to prove Theorem \ref{thm-mean-as}.
\begin{proof}[Proof of Theorem \ref{thm-mean-as}]
 Let us start by the invariance. Given Lemma \ref{lem:V=0}, it is sufficient to prove 
$$V(\hat\rho(t))=0,\,\,\forall t\geq 0 \Leftrightarrow V(\rho(t))=0, \forall t\geq 0\quad a.s.$$  Since $V$ is linear, we have for all $t\geq0$ \begin{eqnarray}
V(\hat\rho(t))&=&\mathbb E(V(\rho(t)))
\end{eqnarray} and then the implication \emph{almost surely} $\Rightarrow$ \emph{in mean} is immediate.

For the opposite direction,  let us remark that $V(\rho(t))\geq 0$ for all $t\geq 0$. This way if we assume that $\mathbb E(V(\rho(t)))=V(\hat\rho(t))=0$ for all $t>0$, it follows that $V(\rho(t))=0$, for all $t\geq0$ almost surely and the result holds.

\medskip
For the GAS property, given Lemma \ref{lem:V=0}, it is sufficient to prove 
\[\lim_{t\to \infty} V(\hat\rho(t))=0\Leftrightarrow \lim_{t\to\infty}V(\rho(t))=0\mbox{ a.s.}\]

 The implication \emph{almost surely} $\Rightarrow$ \emph{in mean} follows from dominated convergence Theorem applied on $V$. Indeed, we have $\lim_{t\to \infty}V(\rho(t))=0$ a.s. and $V(\rho(t))\leq 1$, for all $t\geq0$. It follows $$\lim_{t\to \infty} V(\hat\rho(t))=\lim_{t\to \infty}\mathbb E(V(\rho(t)))=\mathbb E(\lim_{t\to\infty}V(\rho(t)))=0.$$

The opposite direction relies on convergence for positive super martingales. On one hand, the subspace being GAS in mean implies that $\lim_{t\to \infty}\mathbb E(V(\rho(t)))=0$ for any initial state $\rho_0\in \mathcal S(\mathcal H)$. Since $V(\rho(t))\geq0$, this convergence corresponds to a $L^1$ convergence to $0$. 
On the other hand, since $0\leq V(\rho(t)\leq 1$ and given Lemma \ref{lem:V=0}, the process $(V(\rho(t))$ is a positive bounded super martingale. It follows from bounded super martingale convergence Theorem, that this process converges almost surely and in $L^1$ to a random variable $V_\infty$. The uniqueness of the $L^1$ limit implies $V_\infty=0$ almost surely.
\end{proof}
\noindent Given that the two notions of GAS are equivalent, from now on we do not specify to which notion we refer when we say that a subspace is GAS.

\section{Exponential stability}
This section is devoted to the proof of Theorem \ref{thm:expo_conv} Equations \eqref{eq:as_mean_rate_conv} and \eqref{eq:as_as_rate_conv}. That is, we establish almost sure upper bounds on
$$\limsup_{t\rightarrow\infty}\frac{1}{t}\ln(V(\rho(t)).$$
Equation \eqref{eq:as_mean_rate_conv} is proved using a super martingale almost sure convergence while Equation \eqref{eq:as_as_rate_conv} relies on the minimization of a function over a convex set and the strong law of large numbers for square integrable martingales. In Section 5, we further discuss the significance and differences of these two bounds by studying some specific examples. 

\subsection{Preliminaries}
We start by introducing a number of functions that will be instrumental to the proofs. For $\rho\in\mathcal S(\mathcal H), \rho_S\in\mathcal S(\mathcal H_S)$ and $\rho_R\in\mathcal S(\mathcal H_R)$, and for $j=1,\ldots,p$, define:

\begin{align*}
r_j(\rho)=&{\rm tr}[(C_j+C_j^*)\rho],\\
r_{j,S}(\rho_S)=&{\rm tr}[(C_{j,S}+C_{j,S}^*)\rho_S],\\
r_{j,R}(\rho_R)=&{\rm tr}[(C_{j,R}+C_{j,R}^*)\rho_R].
\end{align*}
These play the role of the expectations of the measurement records associated to diffusive processes.
On the other hand, for $j=p+1,\ldots,n$,
\begin{align*}
v_j(\rho)=&{\rm tr}[C_j\rho C_j^*],\\
v_{j,S}(\rho_S)=&{\rm tr}[C_{j,S}^*C_{j,S}\rho_S],\\
v_{j,R}(\rho_R)=&{\rm tr}[C_{j,R}^*C_{j,R}\rho_R].
\end{align*}
These correspond to expectations for jump type measurement record processes.
We define the related vectors,
\begin{eqnarray*}
&\mathbf{r}(\rho)=(r_j(\rho))_{j=1,\ldots,p},\quad \mathbf{r}_S(\rho_S)=(r_{j,S}(\rho_S))_{j=1,\ldots,p},\quad\mathbf{r}_R(\rho_R)=(r_{j,R}(\rho_R))_{j=1,\ldots,p},\\
&\mathbf{v}(\rho)=(v_j(\rho))_{j=p+1,\ldots,n},\quad\mathbf{v}_S(\rho_S)=(v_{j,S}(\rho_S))_{j=p+1,\ldots,n},\quad\mathbf{v}_R(\rho_R)=(v_{j,R}(\rho_R))_{j=p+1,\ldots,n}.
\end{eqnarray*}
In the following, for two vectors $\mathbf a,\mathbf b$, the division $\frac{\mathbf a}{\mathbf b}$ is meant element by elements: $\frac{\mathbf a}{\mathbf b}=(\frac{a_j}{b_j})_j$;  $\mathbf a.\mathbf b$ denotes the Euclidean inner product; for any function $f$ of $\mathbb R$, $f(\mathbf a)=(f(a_j))_j$ and $\|{\bf a}\|$ is the Euclidean norm.

\medskip
We recal Assumption {\bf SP} from Theorem \ref{thm:expo_conv}. While not really restrictive it is essential to our proofs.
\begin{quote}
{\bf Assumption SP:} $C_{j,R}^*C_{j,R}>0$ for all $j=p+1,\ldots,n$.
\end{quote}
This assumption particularly implies that for any $j=p+1,\ldots,n$ and any $\rho\in\mathcal S(\mathcal H)\setminus\mathcal I_S(\mathcal H)$, $v_{j,R}(\rho_R)>0$.

\medskip
The following function is central to the definition of $\beta_0$ in Theorem \ref{thm:expo_conv}.
\begin{align*}
\alpha:\;&\mathcal S(\mathcal H)\times\mathcal S(\mathcal H_R)\to\mathbb R_+\\
	&(\rho,\rho_R)\mapsto\left\{\begin{array}{l} 0\quad\mbox{ if }\exists j=p+1,\ldots,n\mbox{ s.t. }v_{j,R}(\rho_R)=0\\
\frac{1}{2}\|\mathbf{r}(\rho)-\mathbf{r}_R(\rho_R)\|^2 + (\mathbf{v}_R(\rho_R)-\mathbf{v}(\rho)).\mathbf{1}+\mathbf{v}(\rho).\ln\left(\frac{\mathbf{v}(\rho)}{\mathbf{v}_R(\rho_R)}\right)\mbox{ else,}\end{array}\right.
\end{align*}
with the convention $x\ln(x)=0$ whenever $x=0$ and $\mathbf{1}=(1)_{j=p+1,\ldots,n}$. Given that definition we have:

\begin{lem}\label{lem:alpha continuous}
Provided assumption {\bf SP} is fulfilled, $\alpha$ is continuous on $\mathcal S(\mathcal H)\times\mathcal S(\mathcal H_R)$ and the following minimum is well defined:
\begin{align*}
\alpha_1=&\min\{\alpha(\rho,\rho_R)\,|\,\rho\in\mathcal I_S(\mathcal H), \rho_R\in\mathcal S(\mathcal H_R)\}.
\end{align*}

\end{lem}
\begin{proof} Let introduce the set 
$$A=\{(\rho,\rho_R)\in\mathcal S(\mathcal H)\times\mathcal S(\mathcal H_R)\vert \exists j=p+1,\ldots,n,v_{j,R}(\rho_R)=0\}.$$
The set $A$ corresponds to the set of possible points of discontinuity for the function $\alpha$. By definition $\alpha=0$ on $A$. Nevertheless under the assumption {\bf SP}, since $\mathcal S(\mathcal H_R)$ is compact, we get that
\[\min_{j=p+1,\ldots,n}\min_{\rho_R\in\mathcal S(\mathcal H_R)} v_{j,R}(\rho_R)>0.\]
It follows that $A$ is empty and that $\alpha$ is continuous. Since the underlying set is compact and since $\alpha$ is continuous, the minimum is well defined.
\end{proof}

%

\bigskip
We turn to the different exponents that leads to the proofs of Theorem \ref{thm:expo_conv}. 

\noindent Recall that
\begin{itemize}
\item $-\alpha_0$ is the eigenvalue of $\mathcal {L}_R$ with minimum real part,
\item $\alpha_1$ is given in Lemma \ref{lem:alpha continuous},
\end{itemize}
 and define
\begin{itemize}
\item$\alpha_0'=\min {\rm spec}\left(\sum_{j=1}^n C_{j,P}^*C_{j,P}\right).$
\end{itemize}
Equations \eqref{eq:as_mean_rate_conv} and \eqref{eq:as_as_rate_conv} of Theorem \ref{thm:expo_conv} are transcribed in the two following points
\begin{itemize}
\item Provided $\mathcal H_S$ is GAS,
\begin{align*}
\limsup_{t\rightarrow+\infty} \frac{1}{t}\ln(V(\rho(t))\leq -\alpha_0\quad \mbox{a.s.}
\end{align*}
\item If moreover assumption {\bf SP} is fulfilled, 
\begin{align}
\limsup_{t\rightarrow+\infty} \frac{1}{t}\ln(V(\rho(t))\leq -(\alpha'_0+\alpha_1)\quad a.s.\label{eq:improved_rate_conv} 
\end{align}
Then putting 
$$\beta_0=\max(\alpha_0,\alpha_0'+\alpha_1)$$
we get Equation \eqref{eq:as_as_rate_conv}:
\begin{align*}
\limsup_{t\rightarrow+\infty} \frac{1}{t}\ln(V(\rho(t))\leq -\beta_0\quad a.s.
\end{align*}
\end{itemize}

Before turning to the proofs we clarify the relationship between $\alpha_0$ and $\alpha_0'$, and give a sufficient condition for $\alpha_1>0$.
\begin{prop}\label{prop:alpha_0' leq alpha_0}
Assume $\mathcal H_S$ is GAS, then
\[0\leq\alpha_0'\leq \alpha_0.\]
\end{prop}
\begin{proof}
Since $\alpha_0'$ is an element of the spectrum of a positive semi definite operator, it is non negative. Then, on the one hand, we have for any $\mu_R\in \mathcal S(\mathcal H_R)$, ${\rm tr}(\mathcal L_R(\mu_R))\leq -\alpha_0'$. Hence by Gronwall's inequality,
\begin{equation}\label{eq:mm}{\rm tr}(e^{t\mathcal L_R}\mu_R)\leq e^{-\alpha_0' t}.
\end{equation}
On the other hand, from completely positive map Perron--Frobenius spectral Theorem \cite{Evans}, there exists $\rho_R\in \mathcal S(\mathcal H_R)$ such that $e^{t\mathcal L_R}\rho_R=e^{-t\alpha_0}\rho_R$.  Then, applying \eqref{eq:mm} with $\rho_R$ we get $e^{-t\alpha_0}\leq e^{-t\alpha'_0}$ which gives the announced inequality.\end{proof}

The following assumption is a necessary and sufficient condition to have $\alpha_1>0$. It is similar to a non degeneracy condition in non demolition measurements \cite{bbqnd,bbbqnd,bpqnd,ballesteros15}.
\begin{quote}{\bf Assumption ND.} For any $\rho_S\in\mathcal S(\mathcal H_S),\rho_R\in\mathcal S(\mathcal H_R)$, there exists $j=1,\ldots,n$ such that 
\[r_{j,S}(\rho_S)\neq r_{j,R}(\rho_R)\mbox{ if }j=1,\ldots,p,\]
or
\[v_{j,S}(\rho_S)\neq v_{j,R}(\rho_R)\mbox{ if }j=p+1,\ldots,n.\]
\end{quote}

\begin{prop}\label{prop:alpha_1 positif}
Assume {\bf SP} is fulfilled. The assumption {\bf ND} is equivalent to
\[\alpha_1>0.\]
\end{prop}
\begin{proof}
We start with the implication {\bf ND} $\Rightarrow \alpha_1>0$. From Lemma \ref{lem:alpha continuous}, since assumption {\bf SP} is provided, $\alpha$ is continuous on $\mathcal I_S(\mathcal H)\times \mathcal S(\mathcal H_R)$, which is a compact set. Thus the minimum is reached for some $(\rho,\rho_R)\in \mathcal I_S(\mathcal H)\times \mathcal S(\mathcal H_R)$. Since for any $\rho\in \mathcal I_S(\mathcal H)$, $\mathbf{r}(\rho)=\mathbf{r}_S(\rho_S)$ and  $\mathbf{v}(\rho)=\mathbf{v}_S(\rho_S)$, it follows from assumption {\bf ND} that there exists at least one $j$ such that $r_{j,S}(\rho_S)\neq r_{j,R}(\rho_R)$ if $j\leq p$ or $v_{j,S}(\rho_S)\neq v_{j,R}(\rho_R)$ if $j> p$.
The functions $(x,y)\mapsto (x-y)^2$ and $(x,y)\mapsto y-x+x\ln(x/y)$ are positive on respectively $\mathbb R^2$ and $\mathbb R_+\times\mathbb R_+\setminus\{0\}$. They vanish if and only if $x=y$. Thus from the definition of $\alpha$, we get $\alpha_1=\alpha(\rho,\rho_R)>0$.

The opposite implication is obtained by contradiction. Assume $\alpha_1>0$ and there exists a couple $(\rho,\rho_R)\in \mathcal I_S(\mathcal H)\times \mathcal S(\mathcal H_R)$ such that $\mathbf{r}(\rho)=\mathbf{r}_R(\rho_R)$ and $\mathbf{v}(\rho)=\mathbf{v}_R(\rho_R)$. Then $\alpha(\rho,\rho_R)=0$ and thus $\alpha_1=0$ which contradicts the assumption $\alpha_1>0$.
\end{proof}

\subsection{Theorem \ref{thm:expo_conv} Equation \eqref{eq:as_mean_rate_conv} proof}
Fix $\epsilon>0$. From Lemma \ref{lem:spectral_radius_majoration}, there exists $K_R\geq I_{\mathcal H_R}$ such that $\mathcal L_R^* K_R\leq -(\alpha_0 -\frac12\epsilon)K_R$. 
From Equation \eqref{eq:L_R semi group},
$$V_K(\hat\rho(t))={\rm tr}(e^{t\mathcal L_R^*}K_R \rho_{R})\leq e^{-(\alpha_0-\frac12\epsilon)t}V_K(\rho(0)).$$
 In terms of the expectation,
 $$\mathbb E(V_K(\rho(t))e^{(\alpha_0-\epsilon)t})\leq V_K(\rho(0))e^{-\frac12\epsilon t}.$$
 It follows that
\[\lim_{t\to\infty} V_K(\rho(t))e^{(\alpha_0-\epsilon)t}=0,\quad \text{ in }L^1-norm.\]

As for the proof of the almost sure convergence of $V(\rho(t))$ to $0$, we show that $(V_K(\rho(t))e^{(\alpha_0-\epsilon)t})_{t\in \mathbb R_+}$ is a positive super martingale. Using that $(\rho(t))$ is a Markov process and that $V_K$ linear, for any $s\leq t$, we get
\begin{align*}
\mathbb E(V_K(\rho(t))|\mathcal F_s)e^{(\alpha_0-\epsilon)t}=&{\rm tr}[e^{(t-s)\mathcal L_R^*}K\rho_R(s)]e^{(\alpha_0 -\epsilon)t}\\ 
		\leq&V_K(\rho(s))e^{-(\alpha_0-\frac12\epsilon)(t-s)}e^{(\alpha_0-\epsilon) t}\\
		\leq&V_K(\rho(s))e^{(\alpha_0-\epsilon)s}e^{-\frac12\epsilon (t-s)}\\
		\leq&V_K(\rho(s))e^{(\alpha_0-\epsilon)s},
\end{align*}
then $(V_K(\rho(t))e^{(\alpha_0-\epsilon)t})_{t\in \mathbb R_+}$ is a positive super martingale. It follows that this super martingale converges almost surely to a random variable denoted by $Z$. Now, using the fact that $L^1$ convergence implies almost sure convergence for an extracted subsequence we can conclude that $Z=0$. 

Now since $K\geq I_{\mathcal H_R}$, $V(\rho)\leq V_K(\rho)$ for any $\rho\in\mathcal S(\mathcal H)$. Then
\[\lim_{t\to\infty}V(\rho(t))e^{(\alpha_0-\epsilon)t}=0\quad\mbox{a.s.}\]
Namely there a.s. exists $T$ such that for all $t\geq T$
$$V(\rho(t))\leq e^{-(\alpha_0-\epsilon)t}.$$
This implies
$$\limsup_{t\rightarrow\infty}\frac{1}{t}\ln(V(\rho(t)))\leq-(\alpha_0-\epsilon)\quad a.s$$
and taking then $\epsilon$ going to zero yields Equation \eqref{eq:as_mean_rate_conv}.
\hfill\qed

\subsection{Theorem \ref{thm:expo_conv} Equation \eqref{eq:as_as_rate_conv} proof}

The first step in the proof is noticing that $(V(\rho(t)))$ can be expressed as the solution of a Doleans--Dade equation.

Let us introduce the following notation
\[\mathbf{W}(t)=(W_j(t))_{j=1,\ldots,p}\mbox{ and }\hat{\mathbf{N}}(t)=(\hat N_j(t))_{j=p+1,\ldots,n}.\] Furthermore, let
\begin{align}\label{eq:def rho_R red.}
\rho_{R,{\rm red.}}(t)=&\left\{\begin{array}{ll}\frac{\rho_R(t)}{{\rm tr}(\rho_R(t))}&\mbox{if }{\rm tr}(\rho_R(t))=V(\rho(t))\neq 0\\
												\mu_R&\mbox{if }{\rm tr}(\rho_R(t))=V(\rho(t))=0,\end{array}\right.
\end{align}
where $\mu_R\in\mathcal S(\mathcal H_R)$ is arbitrary.

\textbf{Remark:} {In the above definition, the process $\rho_{R,{\rm red.}}(t)$ is a normalized version of the bloc $\rho_R(t)$ obtained by dividing it by ${\rm tr}(\rho_R(t))=V(\rho(t))$. However, since in general nothing ensures that ${\rm tr}(\rho_R(t))$ does not vanish, we introduce an arbitrary state $\mu_R$ in the case ${\rm tr}(\rho_R(t))$ is zero. This is only a formal construction, as the interesting cases are the ones where this quantity is never zero. In fact, by using the invariance property it is easy to see that if $V(\rho(t))=0$ for some $t$, we get $V(\rho(s))=0$ for all $s\geq t$. In this situation the exponential stability is somewhat trivial: the Lyapunov exponent is equal to $-\infty$, as we have convergence to zero in finite time. In the situation where $(V(\rho(t)))$ does not vanish in finite time, the state $\mu_R$ does not actually play any role. The introduction of $\mu_R$ is only instrumental to a proper definition of $\rho_{R,{\rm red.}}(t),$ and the final result, in the cases of interest, will not rely on the choice of $\mu_R$. }

In order to simplify the notation we put
$$V(t):=V(\rho(t)), \forall t\in \mathbb R_+.$$

\begin{prop}\label{prop:explicit V}
The process $(V(t))$ is the unique solution of the SDE
\begin{align}\label{eq:V SDE}
\begin{split}
dV(t)=&V(t-)\Big\{{\rm tr}(\mathcal L_R \rho_{R,\rm red.}(t-))dt\\
		&\quad+ [\mathbf{r}_R(\rho_{R,\rm red.}(t-))-\mathbf{r}(\rho(t-)))].d\mathbf{W}(t)\\
		&\quad+\left(\frac{\mathbf{v}_{R}(\rho_{R,\rm red.}(t-))}{\mathbf{v}(\rho(t-))}-\mathbf{1}\right).[d\hat{\mathbf{N}}(t) - \mathbf{v}(\rho(t-))dt]\Big\},\\
V(0)=&{\rm tr}(P_R\rho_0).
\end{split}
\end{align}
This process is a Doleans--Dade exponential whose explicit expression is
\begin{align}\label{eq:explicit V}
\begin{split}
V(t)=V(0)&\prod_{j=p+1}^n\prod_{s\leq t}\left(1+\left(\frac{v_{j,R}(\rho_{R,\rm red.}(s-))}{v_j(\rho(s-))}-1\right)\Delta\hat N_j(s)\right)\\
	&\times\exp\Big\{\int_0^t{\rm tr}(\mathcal L_R \rho_{R,\rm red.}(s-))ds\\
			&\quad- \frac12\|\mathbf{r}_R(\rho_{R,\rm red.}(s-))-\mathbf{r}(\rho(s-))\|^2ds\\
			& \quad+\int_0^t[\mathbf{r}_R(\rho_{R,\rm red.}(s-))-\mathbf{r}(\rho(s-))].d\mathbf{W}(s)\\
			& \quad- \int_0^t [\mathbf{v}_{R}(\rho_{R,\rm red.}(s-))-\mathbf{v}(\rho(s-))].\mathbf{1}ds\Big\}.
\end{split}
\end{align}

\end{prop}
\begin{proof}
 Since $(\rho(t))$ is well defined, the uniqueness of the solution of  \eqref{eq:V SDE} and the expression \eqref{eq:explicit V} follow from usual arguments of stochastic calculus. Now the fact that $(V(t))$ satisfies \eqref{eq:V SDE} is obtained by applying $V$ on \eqref{eq:def_trajectory2}. Indeed, let us recall that $V(\rho)=0$ implies $\rho_R=0$. This way for all $t\geq0$, we can write $\rho_R(t)=V(t)\rho_{R,\rm red.}(t)$ and applying $V$ (which is linear) on \eqref{eq:def_trajectory2}, we get
\begin{align*}
dV(t)=&{\rm tr}(\mathcal L_R(\rho_R(t-)))dt\\
	& + \sum_{j=1}^p ({\rm tr}((C_{j,R}+C_{j,R}^*)\rho_R(t-))-{\rm tr}((C_j+C_{j,R}^*)\rho(t))V(t-))dW_j(t)\\
	&+\sum_{j=p+1}^n \left(\frac{{\rm tr}(C_{j,R}\rho_R(t-)C_{j,R})}{v_j(t-)}-V(t-)\right)[d\hat N_j(t)-v_j(t-)dt],
\end{align*}
which is the expansion of \eqref{eq:V SDE}.
\end{proof}

The discussion of $(V(t))$ strict positivity is clearer with $V(t)$ written in the following form
\begin{align}\label{eq:V continuous SDE}
\begin{split}
V(t)=&V(0)+\int_0^tV(s)\Big\{{\rm tr}(\mathcal L_R(\rho_R(s-)))ds \\
	&\quad\quad - [\mathbf{v}_R(\rho_{R,\rm red.}(s-))-\mathbf{v}(\rho(s-))].\mathbf{1}dt \\
	&\quad\quad+[\mathbf{r}_R(\rho_{R,\rm red.}(s-))-\mathbf{r}(\rho(s-))].d\mathbf{W}(s)\Big\}\\
	&\quad\quad+\sum_{n=0}^{+\infty}V(T_n-)\frac{v_{j_{T_n},R}(\rho_{R,\rm red.}(T_n-))}{v_{j_{T_n}}(\rho(T_n-))}\mathbf 1_{T_n\leq t},
\end{split}
\end{align}
where the sequence of stopping time $(T_n)$ is defined by $T_0=0$ and $$T_{n+1}=\inf\{t> T_n \mbox{ s.t. } \hat{\mathbf{N}}(t).\mathbf{1}\geq n\}.$$ Note that the independence of $N_i$ ensures that for all $n$, we have $\hat{\mathbf{N}}(T_n).\mathbf{1}=n$ almost surely (two jumps can not appear at the same time).

\textbf{Remark:} Under the light of the expression \eqref{eq:explicit V}, one can introduce $\tau=\inf\{t>0\ /\ V(t)=0\}$. Using strong Markov property of the couple $(\rho(t),V(t))$, $V(t)=0$ for all $t\geq\tau$.  
The following corollary expresses that under assumption {\bf SP} the event $\{\tau<\infty\}$ is of probability $0$.
\begin{cor}\label{cor:V>0}
Assume {\bf SP} is fulfilled. Then for all $t\in\mathbb R_+$, $V(\rho(t))>0$ almost surely.
\end{cor}
\begin{proof}
Assumption {\bf SP} ensures there exists $c>0$ such that for all $j=p+1,\ldots,n$ and any $\rho_R\in\mathcal S(\mathcal H_R), v_{j,R}(\rho_R)>0$. It follows that $\frac{v_{j,R}(\rho_{R,{\rm red.}}(s-))}{v_j(\rho(s-))}>0$ almost surely for all $s\geq 0$ and $j=p+1,\ldots,n$. Thus equation \eqref{eq:explicit V} or \eqref{eq:V continuous SDE}, imply that $V(t)$ does not vanish when a jump occurs (that is at a time $T_n$). Concerning the smooth evolution (that is the diffusive evolution in between the jumps, i.e $\Delta N_i(.)=0$) equation  \eqref{eq:explicit V} implies that $V$ is an exponential and thus does not vanish. 
\end{proof}

The following technical lemma will be used in the next proposition. The proof is based on an argument regarding strong law of large numbers for martingales. We refer the reader to any introductory textbook or lecture on martingale theory for the proof.
\begin{lem}\label{lem:martingales}
Let $\mathbf{F}_{W}:\mathcal S(\mathcal H)\to \mathbb R^p$ be a bounded function. Let $\mathbf{F}_{J}:\mathcal S(\mathcal H)\to \mathbb R^{n-p}$ be a function such that $\rho\mapsto\mathbf{v}(\rho).\mathbf{F}_J^2(\rho)$ is bounded.  Then the processes $(M_W(t))$ and $(M_J(t))$ defined by
\begin{align}
M_W(t)=&\int_0^t\mathbf{F}_W(\rho(s-)).d\mathbf{W}(s)\label{eq:def M_W},\\
M_J(t)=&\int_0^t \mathbf{F}_J(\rho(s-)).[d\hat{\mathbf{N}}(s)-\mathbf{v}(\rho(s-))ds]\label{eq:def M_J}
\end{align}
are square integrable martingales that obey the strong law of large numbers:
\begin{align}
\lim_{t\to\infty}\frac1t M_W(t)&=0\label{eq:LLN M_W},\\
\lim_{t\to\infty}\frac1t M_J(t)&=0\label{eq:LLN M_J}
\end{align}
almost surely.
\end{lem}

\noindent {In the following lemma we use a definition similar to the one of $\rho_{R,\rm red.}$  for a reduced state $\rho_{S,\rm red.}$ on $\mathcal S(\mathcal H_S):$}
\begin{align}\label{eq:def rho_S red.}
\rho_{S,{\rm red.}}=&\left\{\begin{array}{ll}\frac{\rho_S}{{\rm tr}(\rho_S)}&\mbox{if }{\rm tr}(\rho_S)\neq 0\\
												\mu_S&\mbox{if }{\rm tr}(\rho_S)=0.\end{array}\right.
\end{align}
\begin{lem}\label{lem:conv to alpha}
Assume $\mathcal H_S$ is GAS and {\bf SP} is fulfilled. Then,
\begin{align*}
\lim_{t\to \infty} \alpha(\rho(t),\rho_{R,\rm red.}(t))- \alpha\left(\left(\begin{array}{cc}\rho_{S,\rm red.}(t)&0\\0&0\end{array}\right),\rho_{R,\rm red.}(t)\right)=0
\end{align*}
almost surely.
\end{lem}
\begin{proof}
From the GAS property, we have $$\lim_{t\to\infty}\rho_R(t)=0\quad \textrm{and}\quad \lim_{t\to\infty}\rho_P(t)=0\quad a.s.$$ Hence $1-V(\rho(t))$ converges almost surely to $1$ and then $$\lim_{t\to\infty}\|\rho_S(t)-\rho_{S,\rm red.}(t)\|=0\quad  a.s.$$ We then have 
\[\lim_{t\to\infty}\left\|\rho(t)-\left(\begin{array}{cc}\rho_{S,\rm red.}(t)&0\\0&0\end{array}\right)\right\|=0\]
 almost surely. The result then follows from the continuity of $\alpha$  ensured by assumption {\bf SP} and Lemma \ref{lem:alpha continuous}.
\end{proof}

We are now in position to prove Equation \eqref{eq:as_as_rate_conv} of Theorem \ref{thm:expo_conv} which, in regard with \eqref{eq:as_mean_rate_conv}, reduces to a proof of:
\begin{align*}
\limsup_{t\to\infty}\frac1t \ln(V(\rho(t)))\leq& -(\alpha_0'+\alpha_1)\quad a.s.
\end{align*}

\begin{proof}[Theorem \ref{thm:expo_conv} equation \eqref{eq:as_as_rate_conv} proof.]
If $\rho_R(0)=0$, the result is trivial. We therefore prove the result only for $\rho_R(0)\neq 0$. Since {\bf SP} is fulfilled, Corollary \ref{cor:V>0} ensures $V(\rho(t))>0$ for all $t\in\mathbb R_+$ almost surely.

Let
$$\mathbf{F}_W(\rho)=\mathbf{r}_R(\rho_{R,\rm red.}) - \mathbf{r}(\rho)\quad\textrm{and}\quad\mathbf{F}_J(\rho)=\ln\left(\frac{\mathbf{v}_R(\rho_{R,\rm red.})}{\mathbf{v}(\rho)}\right).$$  
Both functions fulfill the assumptions of Lemma \ref{lem:martingales}. Using It\^o--L\'evy Lemma for the logarithm function or using the explicit expression of Proposition \ref{prop:explicit V}, we can express $V(t)$ as
\begin{align*}
V(t)=V(0)\times&\exp\Big\{\int_0^t {\rm tr}(\mathcal L_R\rho_{R,\rm red.}(s-)) - \alpha(\rho(s-),\rho_{R,\rm red.}(s-))ds \\
				&\quad\quad+ M_W(t) +M_J(t)\Big\}.
\end{align*}
where
\[M_W(t)=\int_0^t \mathbf{F}_W(\rho(s-)).d\mathbf{W}(s)\]
and
\[M_J(t)=\int_0^t\mathbf{F}_J(\rho(s-)).[d\hat{\mathbf{N}}(s) - \mathbf{v}(s-)ds].\]
are square integrable martingales. At this stage, we have
\begin{eqnarray*}\frac{1}{t}\ln(V(t))&=&\frac{1}{t}\ln(V(0))\\&&+\frac{1}{t}\int_0^t {\rm tr}(\mathcal L_R\rho_{R,\rm red.}(s-)) - \alpha(\rho(s-),\rho_{R,\rm red.}(s-))ds\\&&+ \frac{1}{t}M_W(t) +\frac{1}{t}M_J(t).
\end{eqnarray*}
Now, the strong law of large numbers of Lemma \ref{lem:martingales} implies $$\lim_{t\to\infty}\frac1tM_W(t)=\lim_{t\to\infty}\frac1tM_J(t)=0\quad a.s.$$
Obviously
$$\lim_{t\to\infty}\frac1t \ln(V(0))=0.$$
It remains to treat the integral term. To this end, from the definition of $\alpha_0'$, recall that
\[{\rm tr}(\mathcal L_R\rho_{R,\rm red.}(s-)) \leq -\alpha_0'\] 
 for all $s\in\mathbb R_+$ almost surely. Then from Lemma \ref{lem:conv to alpha} and the definition of $\alpha_1$, we have
\[\limsup_{t\to\infty}-\alpha(\rho(t),\rho_{R,\rm red.}(t))\leq -\alpha_1\]
almost surely.

From the implication \[\limsup_{t\to\infty}f(t)\leq C\Rightarrow\limsup_{t\to\infty}\frac1t\int_0^tf(s)ds\leq C,\]
we finally obtain
\[\limsup_{t\to\infty} \frac1t \ln(V(\rho(t)))\leq -(\alpha_0'+\alpha_1).\]
\end{proof}

\subsection{Corollary \ref{thm:littleo} proof}
The implication Theorem \ref{thm:expo_conv} $\implies$ Corollary \ref{thm:littleo} follows from next lemma.
\begin{lem}\label{lem:littleo}
\begin{enumerate}[(1)]
\item\label{it:mean_upper_Lyap} Assume there exists $c>0$ such that,
\begin{align}\label{eq:upperbound_mean_lyap}
\limsup_{t\to\infty}\frac1t\ln V(\hat \rho(t))\leq -c
\end{align}
then for any $\epsilon>0$,
$$\|\rho(t)-P_S\rho(t)P_S\|=o(e^{-(\frac c2-\epsilon)t}),\quad \text{in }L^1-\text{norm}.$$
\item\label{it:as_upper_Lyap} Assume there exists $c>0$ such that,
\begin{align}\label{eq:upperbound_as_lyap}
\limsup_{t\to\infty}\frac1t\ln V(\rho(t))\leq -c,\quad a.s.
\end{align}
then for any 
$$\|\rho(t)-P_S\rho(t)P_S\|=o(e^{-(\frac c2-\epsilon)t}),\quad a.s..$$
\end{enumerate}
\end{lem}
\begin{proof}
We prove it assuming $\|\cdot\|$ is the Max norm on $\mathcal B(\mathcal H)$. The Lemma generalizes then to any matrix norm by the equivalence of norms on finite dimensional vector spaces.

Using the triangle inequality, we have
\[\left\|\rho-P_S\rho P_S\right\|\leq \|\rho_P\| + \|\rho_Q\| + \|\rho_R\|.\]
Since $\rho\geq 0$, $\rho_Q=\rho_P^*$, $\|\rho_P\|=\|\rho_Q\|$, and Cauchy--Schwarz Theorem implies the inequality
\[\|\rho_P\|^2\leq \|\rho_S\|\|\rho_R\|\]
From the bound $\|\rho_S\|\leq 1$, it follows
\[ \left\|\rho-P_S\rho P_S\right\|\leq 2\|\rho_R\|^{1/2} +  \|\rho_R\|.\]
Since $V(\rho)$ is the trace norm of $\rho_R$ and the Max norm is smaller than the trace norm, $\|\rho_R\|\leq V(\rho)\leq 1$ and it follows,
\begin{equation}\label{eq:inequality_norm_V}
\left\|\rho-P_S\rho P_S\right\|\leq 3 \sqrt{V(\rho)}.
\end{equation}
Since \eqref{eq:upperbound_mean_lyap} implies that for any $\epsilon>0$, $\lim_{t\to\infty}V(\hat \rho(t))e^{(c-2\epsilon)t}=0$, the positivity of $V$ implies $\lim_{t\to\infty}V(\rho(t))e^{(c-2\epsilon)t}=0$ in $L^1$. Then \eqref{eq:inequality_norm_V} yields (\ref{it:mean_upper_Lyap}).

Similarly \eqref{eq:upperbound_as_lyap} implies that for any $\epsilon>0$, $\lim_{t\to\infty}V(\rho(t))e^{(c-2\epsilon)t}=0$ almost surely. Then \eqref{eq:inequality_norm_V} yields (\ref{it:as_upper_Lyap}).
\end{proof}

\section{Improved stability consequences}\label{sec:simulations}
As stated in the introduction, one can taylor examples where $\alpha_0<\beta_0$. Actually, in this section, we show that it is possible to add a measurement channel that modifies neither $\mathcal L_S$ nor $\mathcal L_R,$ yet makes $\alpha_1$ arbitrarily large.
Define 
$$C_{n+1}=\ell_S P_S+\ell_RP_R,$$ 
with $\ell_S,\ell_R\in\mathbb C$.

This new operator accounts for the addition of a diffusive ``non demolition'' measurement, distinguishing wether the state is in $\mathcal H_S$ or $\mathcal H_R$ \cite{bbbqnd,bpqnd}. It is worth noticing that this addition does not modify the invariance and GAS property of $\mathcal H_S$.

Let us introduce the new operator valued functions associated to the SME which includes $C_{n+1}$. We denote them with a  ``\~{}'' in order to distinguish them from the original ones. For any $\rho\in\mathcal S(\mathcal H),$ we have:
\begin{align*}
\tilde{\mathcal L}(\rho)=&\mathcal L(\rho)+C_{n+1}\rho C_{n+1}^* -\frac12\{C^*_{n+1}C_{n+1},\rho\},\\
\mathcal G_{n+1}(\rho)=&C_{n+1}\rho+\rho C_{n+1}^* -{\rm tr}(C_{n+1}+C_{n+1}^*)\rho)\rho.
\end{align*}
Direct computations yield $\tilde{\mathcal L}_S=\mathcal L_S$ and $\tilde{\mathcal L}_R=\mathcal L_R$. Therefore $\tilde\alpha_0=\alpha_0$ and $\tilde\alpha_0'=\alpha_0'$. We only expect $\tilde\alpha_1\neq\alpha_1$.
The new quantum trajectory $(\tilde\rho(t))$ is the solution of the SDE
\begin{align}\label{eq:def rho tilde}\begin{split}
\tilde\rho(t)=&\rho_0+\int_0^t\tilde{\mathcal L}(\tilde\rho(s-))ds\\
	&+\sum_{j=0}^p\int_0^t\mathcal G_j(\tilde\rho(s-))dW_j(s) + \int_0^t \mathcal G_{n+1}(\tilde\rho(s-))dW_{n+1}(s)\\
&+\sum_{j=p+1}^n\int_0^t\int_{\mathbb R}\left(\frac{\mathcal J_j(\tilde\rho(s-))}{v_j(\tilde\rho(s-))}-\tilde\rho(s-)\right)1_{0<x<v_j(\tilde\rho(s-))}[N_j(dx,ds)-dxds]
\end{split}\end{align}
with $(W_{n+1}(t))$ a Wiener process independent of all the other Wiener and Poisson processes\footnote{We define the new filtered probability space $(\tilde\Omega,\tilde{\mathcal F}, (\tilde{\mathcal F}_t), \tilde{\mathbb P}),$ similarly to the original one.}.

If we assume {\bf SP}, then from the definition of $\alpha$ and the corresponding $\tilde\alpha$, we have
\[\tilde\alpha_1=\alpha_1+\frac12{\rm Re}^2(\ell_S-\ell_R).\]
It is then clear that we can play with the value ${\rm Re}^2(\ell_S-\ell_R)$ to increase arbitrarily the value of $\alpha_1$.  Next proposition expresses this fact and follows directly from Theorem \ref{thm:expo_conv}.
\begin{prop}
Assume $\mathcal H_S$ is GAS and {\bf SP} is fulfilled. Then
\[\limsup_{t\to\infty}\frac1t \ln\big(\mathbb E(V(\tilde\rho(t)))\big)\leq -\alpha_0\]
and for any $C>0$ there exists $\ell_S,\ell_R\in\mathbb C$ such that
\[\lim_{t\to\infty}\frac1t \ln\big(V(\tilde\rho(t))\big)\leq -C\quad \text{a.s.}\]
\end{prop}

Hence, whatever is the mean stability exponent $\alpha_0$, we may have an arbitrarily large almost sure asymptotic stability exponent. 
\bigskip

In the particular case of qubits, i.e. two dimensional Hilbert spaces $\mathcal H$, a finer result holds. The inequalities of Theorem \ref{thm:expo_conv} become equalities, showing that the above bound is, in some sense, sharp. 
Note that, in the qubit case, $\mathcal H_S$ and $\mathcal H_R$ have both dimension one. They correspond to two orthogonal projective rays of $\mathcal H$.
The quantum trajectory can then be expressed, in the orthonormal basis associated to $\mathcal H_S$ and $\mathcal H_R,$ as
\[\rho(t)=\left(\begin{array}{cc} p(t)&c(t)\\ \overline{c}(t)&1-p(t)\end{array}\right)\]
for any time $t$. The evolution of $(\rho(t))$ is then uniquely determined by that of $(p(t))$ and $(c(t))$. In particular we have
$$V(\rho(t))=1-p(t),$$
for all $t\geq0$.

For the sake of simplicity we just focus on the case where only two diffusive measurements are involved ($n=p=1$), associated to operators:
\[ C_0=\left(\begin{array}{cc} 0&\ell_P\\0&0\end{array}\right)\text{, }C_1=\left(\begin{array}{cc} \ell_S&0\\0&\ell_R\end{array}\right)\text{ and }H=0,\]
with $\ell_S,\ell_R,\ell_P\in\mathbb C$. This restriction is intended mainly to improve the readability of our proof, as the results extend easily to more general choices of $C_0$ and $C_1$. Also, adding more diffusive measurements or counting measurements is straightforward.

We can translate the SDE \eqref{eq:def_trajectory} to a SDE involving $(p(t),c(t))$. If $(\rho(t))$ is the solution of \eqref{eq:def_trajectory} with $p=n=1$ and the above defined $C_0$ and $C_1$ then its corresponding process $(p(t),c(t))$ is the solution of

\begin{align}
\begin{split}
dp(t)=&(1-p(t))|\ell_P|^2dt \\
&\quad+ 2(1-p(t)){\rm Re}(\ell_P \overline{c}(t))dW_0(t) \\
&\quad+ 2p(t)(1-p(t)){\rm Re}(\ell_S-\ell_R)dW_1(t)
\end{split}\label{eq:dp}\\
\begin{split}
dc(t)=&-\frac12|\ell_P|^2 c(t)dt -\frac12(|\ell_S|^2+|\ell_R|^2 - 2 \ell_S\overline{\ell}_R) c(t) dt\\
&\quad +\big((1-p(t))\ell_P -2c(t){\rm Re}(\ell_P \overline{c}(t))\big)dW_0(t)\\
&\quad  +\big(\ell_S c(t) + \overline{\ell}_R\overline{c}(t) - 2c(t)(p(t){\rm Re}(\ell_S) + (1-p(t)){\rm Re}(\ell_R))\big)dW_1(t).
\end{split}\label{eq:dc}
\end{align}
From equation \eqref{eq:dp} and the definitions of $\alpha_0$, $\alpha_0'$ and $\alpha_1$, we immediately have
\[\alpha_0=\alpha_0'=|\ell_P|^2\text{ and } \alpha_1=2{\rm Re}^2(\ell_S-\ell_R)\]
and the following refinement of Theorem \ref{thm:expo_conv} holds.
\begin{thm}
Consider the two-dimensional system described above. Assume $\ell_P\neq 0$ and $p_0<1$. Then,
\[\lim_{t\to\infty} \frac1t\ln(1-p(t))=-(\alpha_0 + \alpha_1)\text{ a.s.}\]
\end{thm}

\begin{proof}
From \^Ito lemma, we have
\[\ln(1-p(t))=\ln(1-p_0) -\alpha_0 t -\alpha_1\int_0^t p(s) ds - \int_0^t {\rm Re}(\ell_P \overline{c}(s))ds +M_t\]
with $M_t$ a square integrable martingale such that $\lim_{t\to\infty}M_t/t=0$ almost surely. 
Note that $\ell_P\neq 0$ implies the almost sure convergences of $p(t)$ to $1$ and of $c(t)$ to $0$. Adapting the proof of Theorem \ref{thm:expo_conv} Equation \eqref{eq:as_as_rate_conv} yields the result.
\end{proof}
The almost sure convergence towards $\mathcal H_S$ was already known\cite{pellegrini5}, hence the new result in this case is the stability rate derivation. 

 \begin{figure}[H]
\center
\includegraphics[width=.49\textwidth,bb=0 0 576 432]{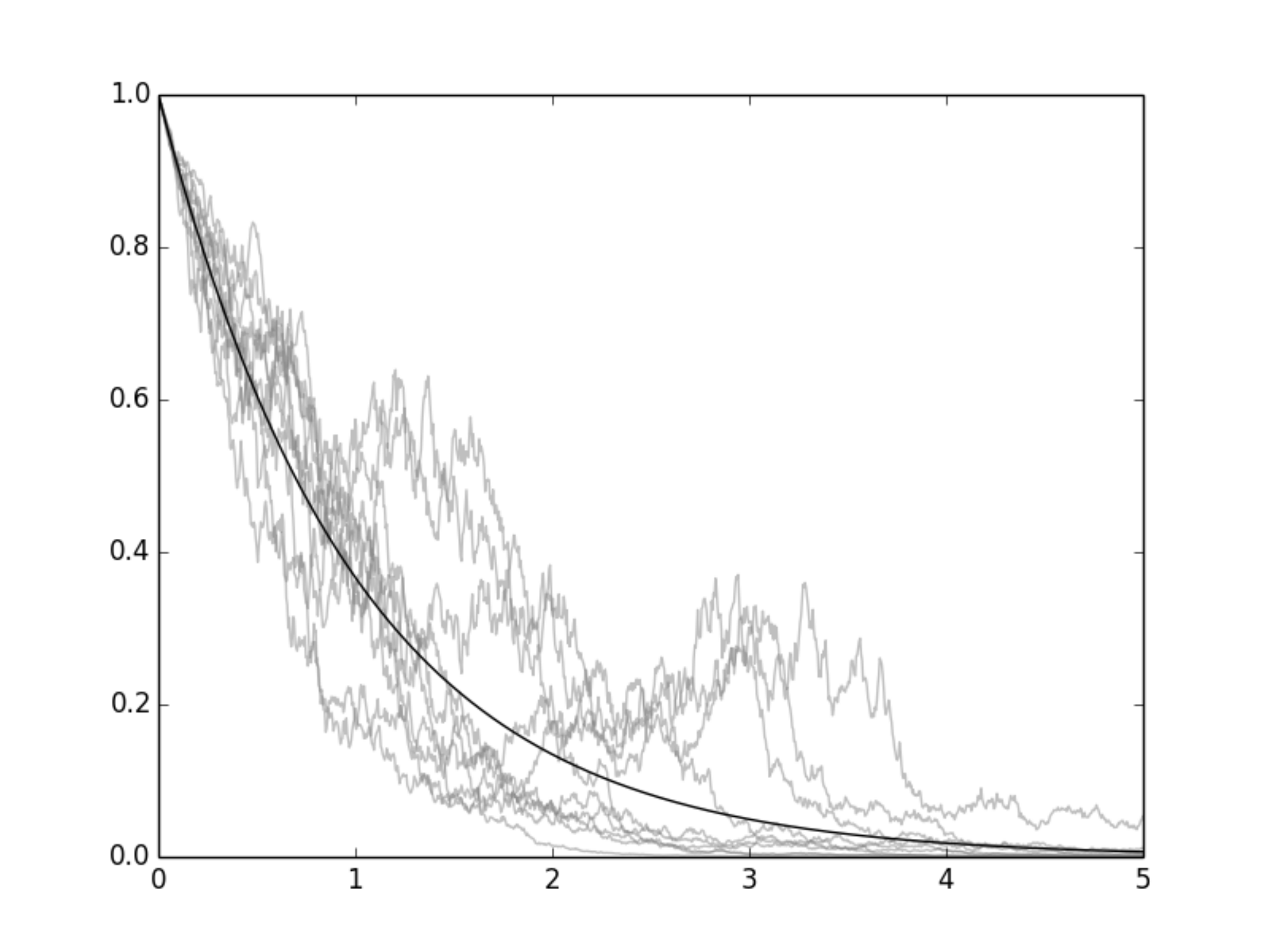}
\includegraphics[width=.49\textwidth,bb=0 0 576 432]{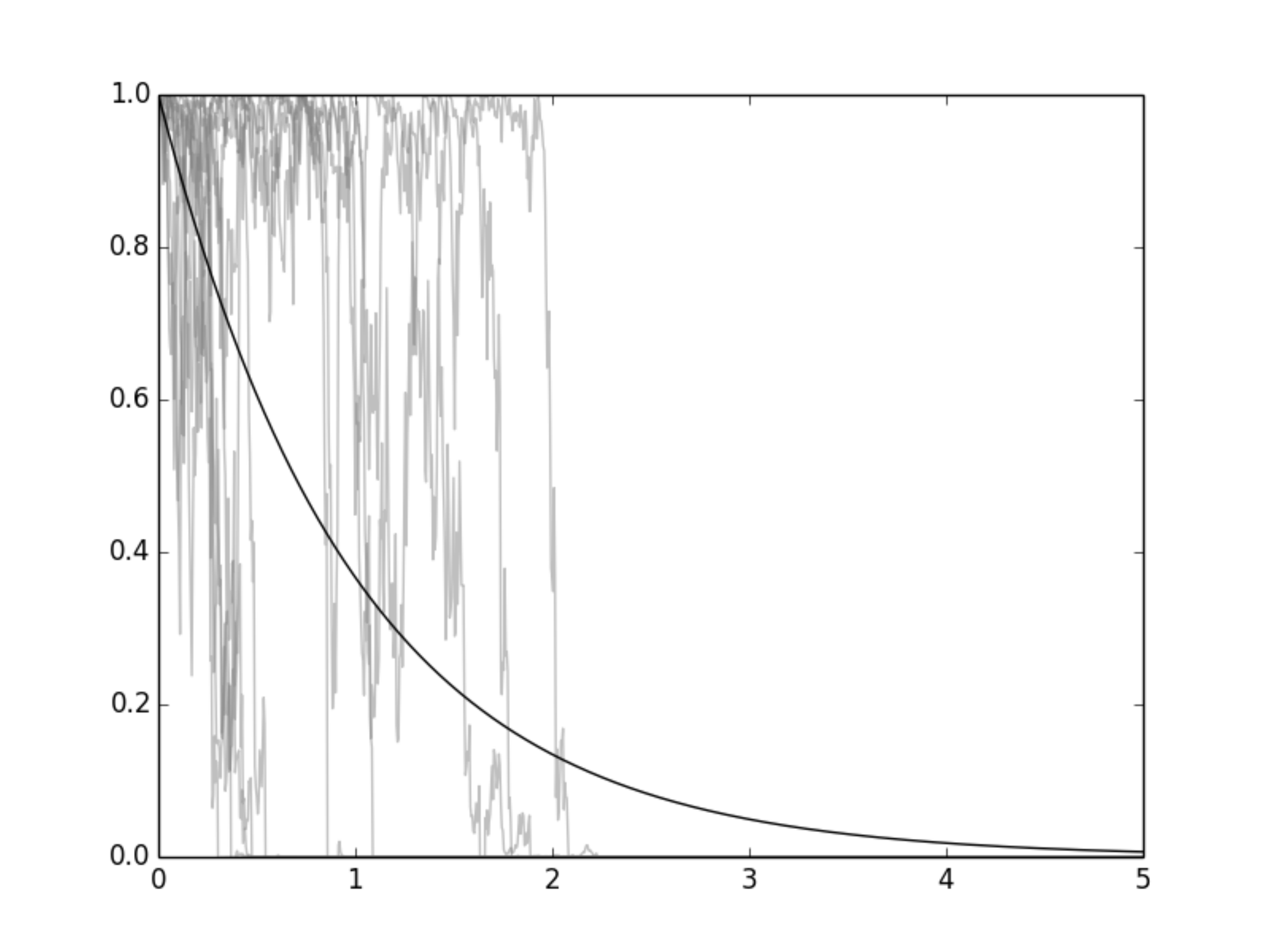}
 \caption{Numerical simulations of the evolution of $1-p(t)$. On the left $\alpha_0=1$ and $\alpha_1=1/2$. On the right $\alpha_0=1$ and $\alpha_1=8$. In each graph one gray line corresponds to a realisation and the solid black line corresponds to the average evolution. The initial condition is set to $p_0=0$. One can remark that when $\alpha_1$ increases, the asymptotic stability increases. \label{fig:simu}} 
 \end{figure}
 
We conclude this section and this article with some numerical simulations (see Figure \ref{fig:simu}) that illustrate the influence of an increased $\alpha_1$ on the typical trajectories.
A larger asymptotic stability rate leads to initially more erratic trajectories, yet the convergence is faster in the sense of the Lyapunov exponents: the increased stability rate makes the state almost ``jump'' to the target subspace, where it remains. This limit behaviour was first remarked and discussed in \cite{bbthermique,bbtjumps1,bbtjumps2,bbtjumps3}. Formulating and proving these observations  more rigorously, \emph{i.e.} studying the limit $\alpha_1\to\infty$, will be the object of further investigation. 
\bigskip

\textbf{Acknowledgments:} The authors thank S. Attal, A. Joye and C.-A. Pillet for the organisation of the summer school in Autrans: `Advances in open quantum systems" where this work has been initiated. We also thank the referee for useful advices regarding the structure of the paper. The research of T. B. was partly supported by ANR
project RMTQIT (grant ANR-12-IS01-0001-01). The research of T. B. and C. P. was partly supported by ANR project StoQ (grant ANR-14-CE25-0003-0). C. P. thanks B. Cloez for stimulating discussions.

\end{document}